\documentclass{llncs}
\usepackage[utf8]{inputenc}

\usepackage{graphicx}
\usepackage{caption}

\usepackage[noadjust]{cite}
\usepackage{enumerate}
\usepackage{thm-restate}
\usepackage{amsmath}
\usepackage{amsfonts}
\usepackage{microtype}

\usepackage{hyperref}

\author{Davide~Bilò\inst{1}\orcidID{0000-0003-3169-4300} \and
Luciano~Gualà\inst{2}\orcidID{0000-0001-6976-5579} \and
Stefano~Leucci\inst{1}\orcidID{0000-0002-8848-7006} \and
Luca~Pepè~Sciarria\inst{2}\orcidID{0000-0003-4432-6099}}

\authorrunning{D. Bilò et al.}

\institute{Department of Information Engineering, Computer Science and Mathematics, University of L'Aquila, Italy, University of L'Aquila, L'Aquila, Italy\\
\email{\{davide.bilo, stefano.leucci\}@univaq.it}
\and
Department of Enterprise Engineering, University of Rome ``Tor Vergata'', University of Rome ``Tor Vergata'', Rome, Italy\\
\email{guala@mat.uniroma2.it, luca.pepesciarria@gmail.com}
}

\title{Finding Diameter-Reducing Shortcuts in Trees}

\newcommand{\diam}{\mathrm{diam}}
\newcommand{\dist}[3]{d_{#1}({#2},{#3})}
\newcommand{\dtie}[3]{{\delta}_{#1}({#2},{#3})}
\newcommand{\T}{\mathcal{T}}
\newcommand{\E}{\mathcal{E}}
\newcommand{\D}{\mathcal{D}}
\newcommand{\lca}{\textsc{lca}}

\usepackage{xspace}
\newcommand{\kdoat}{\textsc{$k$-doat}\xspace}
\newcommand{\doat}{\textsc{doat}\xspace}

\begin{document}

\maketitle

\begin{abstract}
In the \emph{$k$-Diameter-Optimally Augmenting Tree Problem} we are given a tree $T$ of $n$ vertices as input. The tree is embedded in an unknown \emph{metric} space and we have unlimited access to an oracle that, given two distinct vertices $u$ and $v$ of $T$, can answer queries reporting the cost of the edge $(u,v)$ in constant time. We want to augment $T$ with $k$ shortcuts in order to minimize the diameter of the resulting graph. 

For $k=1$, $O(n \log n)$ time algorithms are known both for paths [Wang, CG 2018] and trees [Bilò, TCS 2022]. In this paper we investigate the case of multiple shortcuts. We show that no algorithm that performs $o(n^2)$ queries can provide a better than $10/9$-approximate solution for trees for $k\geq 3$. For any constant $\varepsilon > 0$, we instead design a linear-time $(1+\varepsilon)$-approximation algorithm for paths and $k = o(\sqrt{\log n})$, thus establishing a dichotomy between paths and trees for $k\geq 3$.
We achieve the claimed running time by designing an ad-hoc data structure, which also serves as a key component to provide a linear-time $4$-approximation algorithm for trees, and to compute the diameter of graphs with $n + k - 1$ edges in time $O(n k \log  n)$ even for non-metric graphs.
Our data structure and the latter result are of independent interest.
\end{abstract}

\keywords{Tree diameter augmentation \and Fast diameter computation \and Approximation algorithms \and Time-efficient algorithms}

\section{Introduction}
The \emph{$k$-Diameter-Optimally Augmenting Tree Problem} (\kdoat) is defined as follows. The input consists of a tree $T$ of $n$ vertices that is embedded in an \emph{unknown} space $c$. The space $c$ associates a non-negative cost $c(u,v)$ to each pair of vertices $u,v \in V(T)$, with $u \neq v$. The goal is to quickly compute a set $S$ of $k$ \emph{shortcuts} whose addition to $T$ minimizes the \emph{diameter} of the resulting graph $T+S$.
The diameter of a graph $G$, with a cost of $c(u,v)$ associated with each edge $(u,v)$ of $G$, is defined as $\diam(G) := \max_{u,v \in V(G)} \dist{G}{u}{v}$, where $\dist{G}{u}{v}$ is the distance in $G$ between $u$ and $v$ that is measured w.r.t.\ the edge costs.
We assume to have access to an oracle that answers a query about the cost $c(u,v)$ of any tree-edge/shortcut $(u,v)$ in constant time. 

When $c$ satisfies the triangle inequality, i.e., $c(u,v) \leq c(u,w)+c(w,v)$ for every three distinct vertices $u,v,w \in V(T)$, we say that $T$ is embedded in a {\em metric space}, and we refer to the problem as \emph{metric} \kdoat.

\kdoat and metric \kdoat have been extensively studied for $k=1$.
\kdoat has a trivial lower bound of $\Omega(n^2)$ on the number of queries needed to find $S$ even if one is interested in any finite approximation and the input tree is actually a path.\footnote{As an example, consider an instance $I$ consisting of two subpaths $P_1, P_2$ of $\Theta(n)$ edges each and cost $0$. The subpaths are joined by an edge of cost $1$, and the cost of all shortcuts is $1$.  Any algorithm that does not examine the cost of some shortcut $(u,v)$ with one endpoint in $P_1$ and the other in endpoint in $P_2$ cannot distinguish between $I$ and the instance $I'$ obtained from $I$ by setting $c(u,v)=0$. The claim follows by noticing that there are $\Theta(n^2)$ such shortcuts $(u,v)$ and that the optimal diameters of $I$ and $I'$ are $1$ and $0$, respectively.} On the positive side, it can be solved in $O(n^2)$ time and $O(n \log n)$ space~\cite{Bilo22a}, or in $O(n^2 \log n)$ time and $O(n)$ space~\cite{WangZ21}. Interestingly enough, this second algorithm uses, as a subroutine, a linear time algorithm to compute the diameter of a unicycle graph, i.e., a connected graph with $n$ edges.

Metric \kdoat has been introduced in~\cite{DBLP:journals/corr/GrosseGKSS16} where an $O(n \log^3 n)$-time algorithm is provided for the special case in which the input is a path. In the same paper the authors design a less efficient algorithm for trees that runs in $O(n^2 \log n)$ time. The upper bound for the path case has been then improved to $O(n \log n)$ in~\cite{DBLP:conf/wads/wang}, while in~\cite{Bilo22a} it is shown that the same asymptotic upper bound can be also achieved for trees. Moreover, the latter work also gives a $(1+\varepsilon)$-approximation algorithm for trees with a running time of $O(n+\frac{1}{\varepsilon} \log \frac{1}{\varepsilon})$.

\paragraph*{Our results.}
In this work we focus on metric \kdoat for $k>1$. In such a case one might hope that $o(n^2)$ queries are enough for an exact algorithm. Unfortunately, we show that this is not the case for $3 \le k = o(n)$, even if one is only searching for a $\sigma$-approximate solution with $\sigma < \frac{10}{9}$. Our lower bound is unconditional, holds for trees (with many leaves), and trivially implies an analogous lower bound on the time complexity of any $\sigma$-approximation algorithm.

Motivated by the above lower bound, we focus on approximate solutions and we show two linear-time algorithms with approximation ratios of $4$ and $1+\varepsilon$, for any constant $\varepsilon>0$. The latter algorithm only works for trees with few leaves and $k=o(\sqrt{\log n})$. This establishes a dichotomy between paths and trees for $k \geq 3$: paths can be approximated within a factor of $1+\varepsilon$ in linear-time, while trees (with many leaves) require $\Omega(n^2)$ queries (and hence time) to achieve a better than $10/9$ approximation. Notice that this is not the case for $k=1$. We leave open the problems of understanding whether exact algorithms using $o(n^2)$ queries can be designed for $2$-\doat  on trees and \kdoat on paths.

To achieve the claimed linear-time complexities of our approximation algorithms, we develop an ad-hoc data structure which allows us to quickly compute a small set of well-spread vertices with large pairwise distances. These vertices are used as potential endvertices for the shortcuts. Interestingly, our data structure can also be used to compute the diameter of a \emph{non-metric} graph with $n+k-1$ edges in $O(n k\log n)$ time. For $k=O(1)$, this extends the $O(n)$-time algorithm in~\cite{WangZ21} for computing the diameter of a unicycle graph, with only a logarithmic-factor slowdown. We deem this result of independent interest as it could serve as a tool to design efficient algorithms for \kdoat, as shown for $k=1$ in \cite{WangZ21}.

\paragraph*{Other related work.}

The problem of minimizing the diameter of a graph via the addition of $k$ shortcuts has been extensively studied in the classical setting of optimization problems. This problem is shown to be $\mathsf{NP}$-hard~\cite{DBLP:journals/jgt/SchooneBL87}, not approximable within logarithmic factors unless $\mathsf{P}=\mathsf{NP}$~\cite{BiloGP12},  and some of its variants -- parameterized w.r.t.\ the overall cost of the added shortcuts and w.r.t.\ the resulting diameter -- are even W$[2]$-hard~\cite{FratiGGM15, DBLP:journals/dam/GaoHN13}.
As a consequence, the literature has focused on providing polynomial time approximation algorithms for all these variants~\cite{abs-2209-00370, BiloGP12, ChepoiV02, DBLP:conf/swat/DemaineZ10, FratiGGM15, LiMS92}. 
This differs from \kdoat, where the emphasis is on $o(n^2)$-time algorithms.

Finally, variants of $1$-\doat in which one wants to minimize either the radius or the continuous diameter of the augmented tree have been studied~\cite{DBLP:conf/isaac/0001A16a,DBLP:conf/swat/CarufelMS16,DBLP:conf/wads/CarufelGSS17,DBLP:journals/comgeo/JohnsonW21,DBLP:conf/wads/GudmundssonS21,DBLP:conf/isaac/GudmundssonSY21}. 

\paragraph*{Paper organization.} In Section~\ref{sec:lower_bounds} we present our non-conditional lower bound. Section~\ref{sec:diam_computation} is devoted to the algorithm for computing the diameter of a graph with $n+k-1$ edges in time $O(nk \log n)$. This relies on our data structure, which is described in Section~\ref{sec:ds}.
Finally, in Section~\ref{sec:linear_time_algorithms} we provide our linear-time approximation algorithms for metric \kdoat that use our data structure.
The implementation of our data structure and its analysis, along with some proofs throughout the paper, are deferred to the appendix.

\section{Lower Bound}\label{sec:lower_bounds}

This section is devoted to proving our lower bound on the number of queries needed to solve \kdoat on trees for $k \ge 3$. We start by considering the case $k=3$.

\begin{figure}[t]
  \centering
  \includegraphics[scale=1]{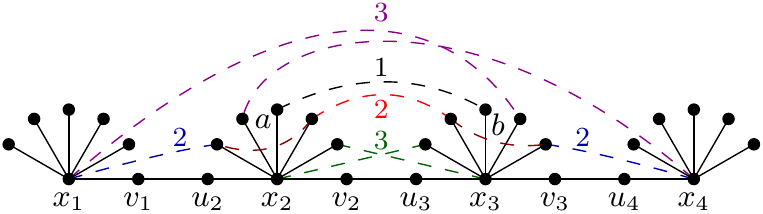}  
 \caption{The graph $G$ of the lower bound construction. The edges of the tree $T$ are solid and have cost $2$; the non-tree edges are dashed and their colors reflect the different types of augmenting edges as defined in the proof of Lemma~\ref{lemma:lb_trees}. To reduce clutter, only some of the augmenting edges are shown.}
 \label{fig:tree_lower_bound}
\end{figure}

\begin{lemma}\label{lemma:lb_trees}
For any sufficiently large $n$, there is a class $\mathcal{G}$ of instances of metric $3$-\doat satisfying the following conditions:
\begin{enumerate}[(i)]
\item In each instance $\langle T, c \rangle \in \mathcal{G}$, $T$ is a tree with $\Theta(n)$ vertices and all tree-edge/shortcut costs assigned by $c$ are positive integers;
\item No algorithm can decide whether an input instance $\langle T, c \rangle$ from $\mathcal{G}$ admits a solution $S$ such that $\diam(T+S) \le 9$ using $o(n^2)$ queries.
\end{enumerate}
\end{lemma}
\begin{proof}
We first describe the $3$-\doat instances. All instances $\langle T, c \rangle \in \mathcal{G}$ share the same tree $T$, and only differ in the cost function $c$.
The tree $T$ is defined as follows: consider $4$ identical stars, each having $n$ vertices and denote by $x_i$ the center of the $i$-th star, by $L_i$ the set of its leaves, and let $X_i =L_i \cup \{x_i\}$. The tree $T$ is obtained by connecting each pair $x_i$, $x_{i+1}$ of centers with a path of three edges $(x_i, v_i), (v_i, u_{i+1})$, and $(u_{i+1}, x_{i+1})$, as in Figure~\ref{fig:tree_lower_bound}. All the tree edges $(u,v)$ have the same cost $c(u,v) = 2$.

The class $\mathcal{G}$ contains an instance $I_{a,b}$ for each pair $(a,b) \in L_2 \times L_3$, and an additional instance $I$.
Fix $(a,b) \in L_2 \times L_3$. The costs of the shortcuts in $I_{a,b}$ are defined w.r.t.\ a graph $G_{a,b}$ obtained by augmenting $T$ with the following edges:
\begin{itemize}
\item All edges $(x_1, y)$ with $y \in L_2$, and all edges $(x_4,y)$ with $y \in L_3$ with cost $2$.
\item The edges $(y,z)$ for every $y \in L_2$ and every $z \in L_3$. The cost $(a,b)$ is $1$, while the cost  of all other edges is $2$;
\item The edges $(y, z)$ for every distinct pair of vertices $y, z$ that are both in $L_2$ or both in $L_3$. The cost of all such edges is $3$.
\item All edges $(x_2,y)$ with $y \in L_3$, and all edges $(x_3,y)$ with $y \in L_2$ with cost $3$; 
\item All edges $(x_1,y)$ with $y \in L_3$, and all edges $(x_4,y)$ with $y \in L_2$ with cost $3$.
\end{itemize}
We define the cost $c(u,v)$ of all remaining shortcuts $(u,v)$ as $\dist{G_{a,b}}{u}{v}$.

We now argue that $c$ satisfies the triangle inequality. Consider any triangle in $G_{a,b}$ having vertices $u,v$, and $w$. We show that the triangle inequality holds for the generic edge $(u,v)$.
As the costs of all edges of $G_{a,b}$, except for the shortcut $(a,b)$, are either $2$ or $3$ and since $c(a,b)=1$, we have that $c(u,v) \leq 3 \leq c(u,w)+c(w,v)$. Any other triangle clearly satisfies the triangle inequality as it contains one or more edges that are not in $G_{a,b}$ and whose costs are computed using distances in $G_{a,b}$. 

To define the cost function $c$ of the remaining instance $I$ of $\mathcal{G}$, choose any $(a,b) \in L_2 \times L_3$, and let $G$ be the graph obtained from $G_{a,b}$ by changing the cost of  $(a,b)$ from $1$ to $2$. We define $c(u,v) = \dist{G}{u}{v}$.
Notice that, in $G$, all edges $(u,v) \in L_2 \times L_3$ have cost $2$ and that the above arguments also show that $c$ still satisfies the triangle inequality.

Since our choice of $I_{a,b}$ and $I$ trivially satisfies (i), we now focus on proving (ii). 

We start by showing the following facts: (1) each instance $I_{a,b}$ admits a solution $S$ such that $\diam(T+S) \le 9$; (2) all solutions $S$ to $I$ are such that $\diam(T+S) \ge 10$; (3) if $u \neq a$ or $v \neq b$ then $\dist{G}{u}{v} = \dist{G_{a,b}}{u}{v}$.

To see (1), consider the set $S=\{(x_1,a),(a,b),(b,x_4)\}$ of $3$ shortcuts. We can observe that $\diam(T+S) \leq 9$. This is because $\dist{T+S}{x_i}{x_j}  \leq 5$ for every two star centers $x_i$ and $x_j$ (see also Figure~\ref{fig:tree_lower_bound}). Moreover, each vertex $u \in L_i$ is at a distance of at most $2$ from $x_i$. Therefore, for every two vertices $u \in X_i$ and $v \in X_j$, we have that 
$\dist{T+S}{u}{v} \leq \dist{T}{u}{x_i} + \dist{T+S}{x_i}{x_j} + \dist{T}{x_j}{v} \leq 2 + 5 + 2 = 9$.

Concerning (2), let us consider any solution $S$ of $3$ shortcuts and define $B_i$ as the set of vertices $X_i$ plus the vertices $u_i$ and $v_i$, if they exist.
We show that if there is no edge in $S$ between $B_{i}$ and $B_{i+1}$ for some $i=1,\dots,3$, then $\diam(T+S) \ge 10$. To this aim, suppose that this is the case, and let $u \in L_i$ and $v \in L_{i+1}$ be two vertices that are not incident to any shortcut in $S$.
Notice that the shortest path in $T+S$ between $u$ and $v$ traverses $x_i$ and $x_{i+1}$, and hence $\dist{T+S}{u}{v} = \dist{T+S}{u}{x_{i}} + \dist{T+S}{x_i}{x_{i+1}} + \dist{T+S}{x_{i+1}}{v} = 4 + \dist{T+S}{x_i}{x_{i+1}}$. We now argue that $\dist{T+S}{x_i}{x_{i+1}} \ge 6$. Indeed, we have that all edges in $T+S$ cost at least $2$, therefore $\dist{T+S}{x_i}{x_{i+1}} < 6$ would imply that the shortest path $\pi$ from $x_i$ to $x_{i+1}$ in $T+S$ traverses a single intermediate vertex $z$. By assumption, $z$ must belong to some $B_j$ for $j \not\in \{i, i+1\}$. However, by construction of $G$, we have $c(x_{i},z) \ge 3$ and $c(x_{i+1}, z) \ge 3$ for every such $z \in B_j$. 

Hence, we can assume that we have a single shortcut edge between $B_i$ and $B_{i+1}$, for $i=1,\dots, 3$. Let $x \in L_1$ (resp. $y \in L_4$) such that no shortcut in $S$ is incident to $x$ (resp. $y$). Notice that every path in $T+S$ from $x$ to $y$ traverses at least $5$ edges and, since all edges cost at least $2$, we have $\diam(T+S) \ge \dist{T+S}{x}{y} \ge 10$.

We now prove (3). Let $u,v$ be two vertices such that there is a shortest path $\pi$ in $G_{a,b}$ from $u$ to $v$ traversing the edge $(a,b)$. We show that there is another shortest path from $u$ to $v$ in $G_{a,b}$ that avoids edge $(a,b)$. Consider a subpath $\pi'$ of $\pi$ consisting of two edges one of which is $(a,b)$. Let $w \neq a,b$ be one of the endvertices of $\pi'$ and let $w' \in \{a,b\}$ be the other endvertex of $\pi'$. Observe that edges $(a,b)$, $(a,w)$, and $(w,b)$ forms a triangle in $G_{a,b}$ and $c(a,w),c(w,b) \in \{2,3\}$. Since $c(a,b)=1$, we have $c(w,a)\leq c(w,b)+c(b,a)$ and $c(w,b) \leq c(w,a)+c(a,b)$. This implies  
that we can shortcut $\pi'$ with the edge $(w,w')$ thus obtaining another shortest path from $u$ to $v$ that does not use the edge $(a,b)$. 

We are finally ready to prove (ii). We suppose towards a contradiction that some algorithm $\mathcal{A}$ requires $o(n^2)$ queries and, given any instance $\langle T, c \rangle \in \mathcal{G}$, decides whether $\diam(T+S) \le 9$ for some set $S$ of $3$ shortcuts.\footnote{For the sake of simplicity, we consider deterministic algorithms only. However, standard arguments can be used to prove a similar claim also for randomized algorithms.}
By (2), $\mathcal{A}$ with input $I$ must report that there is no feasible set of shortcuts that achieves diameter at most $9$. Since $\mathcal{A}$ performs $o(n^2)$ queries, for all sufficiently large values of $n$, there must be an edge $(a,b)$ with $a \in L_2$ and $b \in L_3$ whose cost $c(a,b)$ is not inspected by $\mathcal{A}$.
By (3), the costs of all the edges $(u,v)$ with $(u,v) \neq (a,b)$ are the same in the two instances $I$ and $I_{a,b}$ and hence $\mathcal{A}$ must report that $I_{a,b}$ admits no set of shortcuts $S$ such that $\diam(T+S) \le 9$. This contradicts (1). \qed
\end{proof}

With some additional technicalities we can generalize  Lemma~\ref{lemma:lb_trees} to $3 \le k = o(n)$, which immediately implies a lower bound on the number of queries needed by any $\sigma$-approximation algorithm with $\sigma < 10/9$.
\begin{restatable}{lemma}{lemmalbtreesgeneralized}
\label{lemma:lb_trees_generalized}
For any sufficiently large $n$ and $3 \le k = o(n)$, there is a class $\mathcal{G}$ of instances of metric $k$-\doat satisfying the following conditions:
\begin{enumerate}[(i)]
\item In each instance $\langle T, c \rangle \in \mathcal{G}$, $T$ is a tree with $\Theta(n)$ vertices and all tree-edge/shortcut costs assigned by $c$ are positive integers;
\item No algorithm can decide whether an input instance $\langle T, c \rangle$ from $\mathcal{G}$ admits a solution $S$ such that $\diam(T+S) \le 9$ using $o(n^2)$ queries.
\end{enumerate}
\end{restatable}

\begin{theorem}
    There is no $o(n^2)$-query $\sigma$-approximation algorithm for metric \kdoat with $\sigma < 10/9$ and $3 \leq k =o(n)$.
\end{theorem}

\section{Fast diameter computation}\label{sec:diam_computation}

In this section we describe an algorithm that computes the diameter of a graph on $n$ vertices and $n+k-1$ edges, with non-negative edge costs, in $O(nk\log n)$ time. Before describing the general solution, we consider the case in which the graph is obtained by augmenting a path $P$ of $n$ vertices with $k$ edges as a warm-up.

\subsection{Warm-up: diameter on augmented paths}

 Given a path $P$ and a set $S$ of $k$ shortcuts we show how to compute the diameter of $P + S$. We do that by computing the {\em eccentricity} $\E_s := \max_{v \in V(P)} \dist{P+S}{s}{v}$ of each vertex $s$. Clearly, the diameter of $P+S$ is given by the maximum value chosen among the computed vertex eccentricities, i.e., $\diam(P+S) = \max_{s \in V(P)}\E_s$. Given a subset of vertices $X \subseteq V(P)$, define $\E_s(X) := \max_{v \in X} \dist{P+S}{s}{v}$, i.e., $\E_s(X)$ is the eccentricity of $s$ restricted to $X$.
 
 In the rest of the section we focus on an fixed vertex $s$. We begin by computing a condensed weighted (multi-)graph $G'$. To this aim, we say that a vertex $v$ is \emph{marked as terminal} if it satisfies at least one of the following conditions: (i) $v=s$, (ii) $v$ is an endvertex of some shortcut in $S$, or (iii) $v$ is an endpoint of the path.
 Traverse $P$ from one endpoint to the other and let $v_1, \dots, v_h$ be the marked vertices, in order of traversal. 
 We set the vertex set of $G'$ as the set $M$ of all vertices marked as terminals while the edge set of $G'$ contains (i) all edges $e_i = (v_i, v_{i+1})$ for $i=1,\dots, h-1$, where the cost of edge $e_i$ is $\dist{P}{v_i}{v_{i+1}}$, and (ii) all edges in $S$, with their respective costs. The graph $G'$ has $O(k)$ vertices and edges, and it can be built in $O(k)$ time after a one-time preprocessing of $P$ which requires $O(n)$ time.  See Figure~\ref{fig:shrunk_path} for an example.

We now compute all the distances
from $s$ in $G'$ in $O(k \log k)$ time by running Dijkstra's algorithm. Since our construction of $G'$ ensures that $\dist{P+S}{s}{v} = \dist{G'}{s}{v}$ for every terminal $v$, we now know all the distances $\alpha_i:=\dist{P+S}{s}{v_i}$ with $v_i \in M$. 
 
 For $i=1, \dots, h-1$, define $P_i$ as  the subpath of $P$ between $v_i$ and $v_{i+1}$. In order to find $\E_s(V(P))$, we will separately compute the quantities $\E_s(P_1), \dots, \E_s(P_{h-1})$.\footnote{With a slight abuse of notation we use $P_i$ to refer both to the subpath of $P$ between $v_i$ and $v_{i+1}$ and to the set of vertices therein.}
Fix an index $i$ and let $u$ be a vertex in $P_i$.
Consider a shortest path $\pi$ from $s$ to $u$ in $P+S$ and let $x$ be the last marked vertex traversed by $\pi$ (this vertex always exists since $s$ is marked). We can decompose $\pi$ into two subpaths: a shortest path $\pi_x$ from $s$ to $x$, and a shortest path $\pi_u$ from $x$ to $u$.
 By the choice of $x$, $x$ is the only marked vertex in $\pi_u$. This means that $x$ is either $v_i$ or $v_{i+1}$ and, in particular:
 \[
   \dist{P+S}{s}{u} = \min\left\{ \alpha_i + \dist{P_i}{v_i}{u}, \alpha_{i+1} + \dist{P_i}{v_{i+1}}{u} \right\}.
 \]

 Hence, the farthest vertex $u$ from $s$ among those in $P_i$ is the one that maximizes the right-hand side of the above formula, i.e.:
 \[
    \E_s(P_i) = \max_{u \in P_i} \min\left\{ \alpha_i + \dist{P_i}{v_i}{u}, \alpha_{i+1} + \dist{P_i}{v_{i+1}}{u} \right\}.
 \]
 
 \begin{figure}[t]
    \centering
    \includegraphics[width=\textwidth]{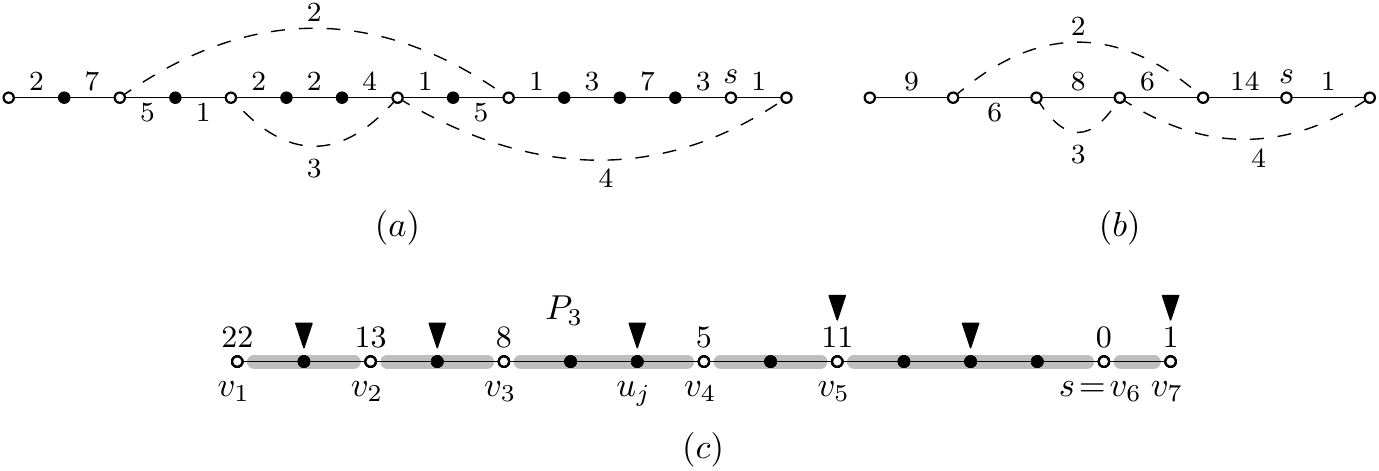}
    \caption{An example showing how to compute the eccentricity of a vertex $s$ in $P+S$. $(a)$ Edges of $P$ are solid, while edges in $S$ are dashed. Vertices marked as terminals are white. $(b)$ The condensed graph $G'$. $(c)$ Every terminal $v_i$ is labeled with its corresponding $\alpha_i$ and the arrows point to the vertices $u_j$ in each (shaded) path $P_i$ such that $j\ge 2$ is the smallest index with $\ell(j) \ge r(j)$.}
    \label{fig:shrunk_path}
\end{figure}
 
 We now describe how such a maximum can be computed efficiently.
 Let $u_j$ denote the $j$-th vertex encountered when $P_i$ is traversed from $v_{i}$ to $v_{i+1}$. The key observation is that the quantity $\ell(j) = \alpha_i + \dist{P_i}{v_i}{u_j}$ is monotonically non-decreasing w.r.t. $j$, while the quantity $r(j) = \alpha_{i+1} + \dist{P_i}{v_{i+1}}{u_j}$ is monotonically non-increasing w.r.t. $j$. Since both $\ell(j)$ and $r(j)$ can be evaluated in constant time once $\alpha_i$ and $\alpha_{i+1}$ are known, we can binary search for the smallest index $j \ge 2$ such that $\ell(j) \ge r(j)$ (see Figure~\ref{fig:shrunk_path}~(c)). Notice that index $j$ always exists since the condition is satisfied for $u_j = v_{i+1}$.

 This requires $O( \log |P_i|)$ time and allows us to return $\E_s(P_i) = \max\{ \ell(j-1), r(j) \}$. 

 After the linear-time preprocessing, the time needed to compute the eccentricity of a single vertex $s$ is then $O(k \log k + \sum_{i=1}^{h-1} \log |P_i|)$. Since $\sum_{i=1}^{h-1}|P_i|<n+h$ and $h = O(k)$, this can be upper bounded by $O(k \log k + k \log \frac{n}{k}) = O(k \log \max \{k, \frac{n}{k}\})$.
 Repeating the above procedure for all vertices $s$, and accounting for the preprocessing time, we can compute the diameter of $P+S$ in time $O(nk \log \max \{k, \frac{n}{k}\})$.

 \subsection{Diameter on augmented trees}

We are now ready to describe the algorithm that computes the diameter of a tree $T$ augmented with a set $S$ of $k$ shortcuts. The key idea is to use the same framework used for the path. Roughly speaking, we define a condensed graph $G'$ over a set of $O(k)$ marked vertices, as we did for the path. Next, we use $G'$ to compute in $O(k \log k)$ time the distance $\alpha_v:=d_{T+S}(s,v)$ for every marked vertex $v$, and then we use these distances to compute the eccentricity $\E_s$ of $s$ in $T+S$. This last step is the tricky one. In order to efficiently manage it, we design an ad-hoc data structure which will be able to compute $\E_s$ in $O(k \log n)$ time once all values $\alpha_v$ are known.
In the following we provide a description of the operations supported by our data structure, and we show how they can be used to compute the diameter of $T+S$. In Appendix~\ref{app:ds} we describe how the data structure can be implemented.

\subsubsection{An auxiliary data structure: description}
\label{sec:ds}

Given a tree $T$ that is rooted in some vertex $r$ and two vertices $u,v$ of $T$, we denote by $\lca_T(u,v)$ their \emph{lowest common ancestor} in $T$, i.e., the deepest node (w.r.t.\ the hop-distance) that is an ancestor of both $u$ and $v$ in $T$.
 Moreover, given a subset of vertices $M$, we define the \emph{shrunk} version of $T$ w.r.t.\ $M$ as the tree $T(M)$ whose vertex set consists of $r$ along with all vertices in $\{ \lca_T(u,v) \mid u,v \in M \}$ (notice that this includes all vertices $v \in M$ since $\lca(v,v) = v$), and whose edge set contains an edge $(u,v)$ for each pair of distinct vertices $u,v$ such that the only vertices in the unique path between $u$ and $v$ in $T$ that belong to $T(M)$ are the endvertices $u$ and $v$.
 We call the vertices in $M$ \emph{terminal} vertices, while we refer to the vertices that are in $T(M)$ but not in $M$ as \emph{Steiner} vertices. See Figure~\ref{fig:shrunk_tree} for an example.
 
 \begin{figure}[t]
     \centering
     \includegraphics[scale=0.8]{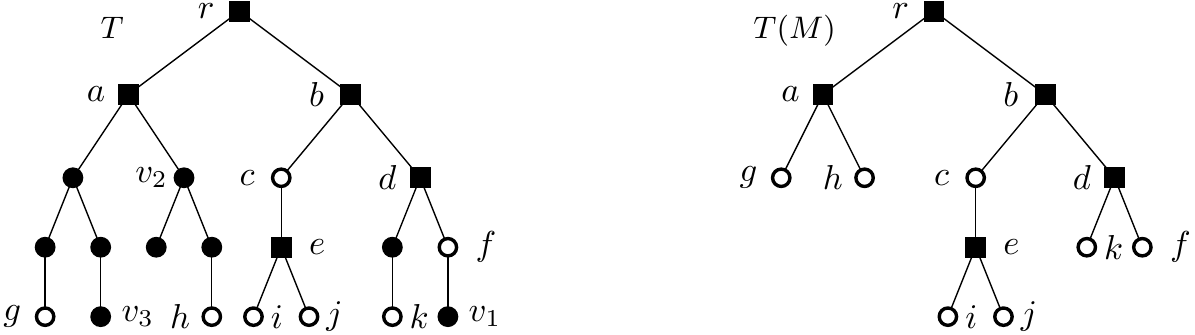}
     \caption{The rooted binary tree $T$ with root $r$ is depicted on the left side. The set of terminal vertices is $M=\{c,f,g,h,i,j,k\}$ and is represented using white vertices. The square vertices are the Steiner vertices. The corresponding shrunk tree $T(M)$ is depicted on the right side.}
     \label{fig:shrunk_tree}
 \end{figure}
 
Our data structure can be built in $O(n)$ time and supports the following operations, where $k$ refers to the current number of terminal vertices:
\begin{description}
    \item[MakeTerminal($v$):] Mark a given vertex $v$ of $T$ as a terminal vertex.  Set $\alpha_v = 0$. This operation requires $O(\log n)$ time.
    
    \item[Shrink()]: Report the shrunk version of $T$ w.r.t.\ the set $M$ of terminal vertices, i.e., report $T(M)$. This operation requires $O(k)$ time.
    
    \item[SetAlpha($v, \alpha$):] Given a terminal vertex $v$, set $\alpha_v = \alpha$. This requires $O(1)$ time. 
    
    \item[ReportFarthest()]: Return $\max_{u \in T} \min_{v \in M} ( \alpha_v + \dist{T}{v}{u}$), where $M$ is the current set of terminal vertices, along with a pair of vertices $(v, u)$ for which the above maximum is attained. This operation requires $O(k \log n)$ time. 
\end{description}
 
\subsubsection{Our algorithm}
In this section we show how to compute the diameter of $T+S$ in time $O(nk \log n)$. We can assume w.l.o.g.\ that the input tree $T$ is binary since, if this is not the case, we can transform $T$ into a binary tree having same diameter once augmented with $S$, and asymptotically the same number of vertices as $T$. 
This transformation requires linear time and is described in Appendix~\ref{sec:binarizzazione}. Moreover, we perform a linear-time preprocessing in order to be able to compute the distance $\dist{T}{u}{v}$ between any pair of vertices in constant time.\footnote{This can be done by rooting $T$ in an arbitrary vertex, noticing that $\dist{T}{u}{v} = \dist{T}{u}{\lca_T(u,v)} + \dist{T}{\lca_T(u,v)}{v}$, and using an \emph{oracle} that can report $\lca(u,v)$ in constant time after a $O(n)$-time preprocessing \cite{HarelT84}.}

We use the data structure $\D$ of Section~\ref{sec:ds}, initialized with the binary tree $T$.
Similarly to our algorithm on paths, we compute the diameter of $T$ by finding the eccentricity $\E_s$ in $T+S$ of each vertex $s$. In the rest of this section we fix $s$ and show how to compute $\E_s$.

We start by considering all vertices $x$ such that either $x=s$ or $x$ is an endvertex of some edge in $S$, and we mark all such vertices as terminals in $\D$ (we recall that these vertices also form the set $M$ of terminals). This requires time $O(k \log n)$.
Next, we compute the distances from $s$ in the (multi-)graph $G'$ defined on the shrunk tree $T(M)$ by (i) assigning weight $\dist{T}{u}{v}$ to each edge $(u,v)$ in $T(M)$, and (ii) adding all edges in $S$ with weight equal to their cost.  This can be done in $O(k \log k)$ time using Dijkstra's algorithm. We let $\alpha_v$ denote the computed distance from $s$ to terminal vertex $v$.

We can now find $\E_s$ by assigning cost $\alpha_v$ to each terminal $v$ in $\D$ (using the SetAlpha($v, \alpha_v$) operation), and performing a ReportFarthest() query. This requires $O(k + k \log n) = O(k \log n)$ time. Finally, we revert $\D$ to the initial state before moving on to the next vertex $s$.\footnote{This can be done in $O(k \log n)$ time by keeping track of all the memory words modified as a result of the operations on $\D$ performed while processing $s$, along with their initial contents. To revert $\D$ to its initial state it suffices to rollback these modifications.} 

\begin{theorem}
Given a graph $G$ on $n$ vertices and $n+k-1$ edges with non-negative edge costs and $k \ge 1$, we can compute the diameter of $G$ in time $O(n k \log n)$.
\end{theorem}

Since there are $\eta = \binom{n}{2} - (n-1) = \frac{(n-1)(n-2)}{2}$ possible shortcut edges, we can solve the \kdoat problem by computing the diameter of $T+S$ for each of the $\binom{\eta}{k} = O(n^{2k})$ possible sets $S$ of $k$ shortcuts using the above algorithm. This yields the following result:

\begin{corollary}\label{cor:exact_algorithm}
The \kdoat problem on trees can be solved in time $O(n^{2k+1} \cdot k \log n)$.
\end{corollary}

\section{Linear-time approximation algorithms}\label{sec:linear_time_algorithms}

In this section we describe two $O(n)$-time approximation algorithms for metric \kdoat. The first algorithm guarantees an approximation factor of $4$ and runs in linear time when $k=O\big(\sqrt{\frac{n}{\log n}}\big)$. The second algorithm computes a $(1+\varepsilon)$-approximate solution for constant $\varepsilon > 0$, but its running time depends on the number $\lambda$ of leaves of the input tree $T$ and it is linear when $\lambda = O\left(n^{\frac{1}{(2k+2)^2}}\right)$ and $k=o(\sqrt{\log n})$.

Both algorithms use the data structure introduced in Section~\ref{sec:ds} and are based on the famous idea introduced by Gonzalez to design approximation algorithms for graph clustering~\cite{Gonzalez85}. Gonzalez' idea, to whom we will refer to as Gonzalez algorithm in the following, is to compute $h$ suitable vertices $x_1,\ldots, x_h$ of an input graph $G$ with non-negative edge costs in a simple iterative fashion.\footnote{In the original algorithm by Gonzalez, these $h$ vertices are used to define the centers of the $h$ clusters.} The first vertex $x_1$ has no constraint and can be any vertex of $G$. Given the vertices $x_1,\ldots,x_i$, with $i < h$, the vertex $x_{i+1}$ is selected to maximize its distance towards $x_1,\ldots,x_i$. More precisely,
$
x_{i+1} \in \arg\max_{v \in V(G)}\min_{1 \leq j \leq i} \dist{G}{v}{x_j}.
$
We now state two useful lemmas, the first of which is proved in~\cite{Gonzalez85}.
\begin{restatable}{lemma}{lemmagonzalez}
\label{lemma:gonzalez}
Let $G$ be a graph with non-negative edge costs and let $x_1,\dots,x_h$ be the vertices computed by Gonzalez algorithm on input $G$ and $h$. Let  $D=\min_{1\leq i < j \leq h}\dist{G}{x_i}{x_j}$. Then, for every vertex $v$ of $G$, there exists $1 \leq i \leq h$ such that $\dist{G}{v}{x_i} \leq D$.
\end{restatable}
\begin{lemma}\label{lemma:gonzalez_fast_implementation}
Given as input a graph $G$ that is a tree and a positive integer $h$, Gonzalez algorithm can compute the vertices $x_1,\ldots, x_h$ in $O(n+h^2 \log n)$ time.
\end{lemma}
\begin{proof}
We can implement Gonzalez algorithm in time $O(n + h^2 \log n)$ by constructing the data structure $\D$ described in Section~\ref{sec:ds} and use it as follows. We iteratively (i) mark the vertex $x_i$ as a terminal (in $O(\log n)$ time), and (ii) query $\D$ for the  vertex $x_{i+1}$ that maximizes the distance from all terminal vertices (in $O(h \log n)$ time). \qed
\end{proof}

In the remainder of this section let $S^*$ be an optimal solution for the \kdoat instance consisting of the tree $T$ embedded in a metric space $c$, and let $D^*=\diam(T+S^*)$.

\subsubsection{The \texorpdfstring{$4$}{4}-approximation algorithm}

The $4$-approximation algorithm we describe has been proposed and analyzed by Li {\em et al.}~\cite{LiMS92} for the variant of \kdoat in which we are given a graph $G$ as input and edge/shortcut costs are uniform. Li {\em et al.}~\cite{LiMS92} proved that the algorithm guarantees an approximation factor of $\big(4+\frac{2}{D^*}\big)$; the analysis has been subsequently improved to $\big(2+\frac{2}{D^*}\big)$ in~\cite{BiloGP12}. We show that the algorithm guarantees an approximation factor of $4$ for the \kdoat problem when $c$ satisfies the triangle inequality.  

The algorithm works as follows. We use Gonzalez algorithm on the input $G=T$ and $h=k+1$ to compute $k+1$ vertices $x_{1},\ldots, x_{k+1}$. The set $S$ of $k$ shortcuts is given by the star centered at $x_1$ and having $x_2,\ldots,x_{k+1}$ as its leaves, i.e., $S=\{(x_1,x_i) \mid 2 \leq i \leq k+1\}$. The following lemma is crucial to prove the correctness of our algorithm.

\begin{lemma}[{\cite[Lemma 3]{FratiGGM15}}]
\label{lemma:frati}
Given $G=T$ and $h=k+1$, Gonzalez algorithm computes a sequence of vertices $x_1,\ldots,x_{k+1}$ with $\min_{1  \leq i \leq k+1}\dist{T}{x_i}{v} \leq D^*$ for every vertex $v$ of $G$.
\end{lemma}

\begin{restatable}{theorem}{thmfourapx}\label{thm:4_apx}
In a \kdoat input instance in which $T$ is embedded in a metric space and $k=O\big(\sqrt{\frac{n}{\log n}}\big)$,
the algorithm computes a $4$-approximate solution in $O(n)$ time.
\end{restatable}

\subsubsection{The \texorpdfstring{$(1+\varepsilon)$}{(1+epsilon)}-approximation algorithm}

We now describe an algorithm that, for any constant $\varepsilon>0$, computes a $(1+\varepsilon)$-approximate solution for metric \kdoat. The running time of the algorithm is guaranteed to be linear when $k=o(\sqrt{\log n})$ and $T$ has $\lambda=O\left(n^{\frac{1}{(2k+2)^2}}\right)$ leaves.

As usual for polynomial-time approximation schemes, we will consider only large enough instances, i.e., we will assume $n\ge n_0$ for some constant $n_0$ depending only on $\varepsilon$. We can solve the instances with $n < n_0$ in constant time using, e.g., the exact algorithm of Section~\ref{sec:diam_computation}.

In particular, we will assume that
\begin{equation}\label{eq:assumption_on_n}
    n > \left(\frac{12\lambda (k+2)^2}{\varepsilon}\right)^{2k+2}.
\end{equation}

Notice that, for any constant $\varepsilon>0$, when $k=o(\sqrt{\log n})$ and $\lambda=O\left(n^{\frac{1}{(2k+2)^2}}\right)$, the right-hand side of \eqref{eq:assumption_on_n} is in $o(n)$. As a consequence, it is always possible to choose a constant $n_0$ such that \eqref{eq:assumption_on_n} is satisfied by all $n \ge n_0$. 

The main idea is an extension of the similar result proved in~\cite{Bilo22a} for the special case $k=1$. However, we also benefit from the fast implementation we provided for Gonzalez algorithm to obtain a linear running time. 

The idea borrowed from~\cite{Bilo22a} is that of reducing the problem instance into a smaller instance formed by a tree $T'$ induced by few vertices of $T$ and by a suitable cost function $c'$. Next, we use the exact algorithm of Corollary~\ref{cor:exact_algorithm} to compute an optimal solution $S'$ for the reduced instance. Finally, we show that $S'$ is a $(1+\varepsilon)$-approximate solution for the original instance. The quasi-optimality of the computed solution comes from the fact that the reduced instance is formed by a suitably selected subset of vertices that are close to the unselected ones.

The reduced instance is not exactly a \kdoat problem instance, but an instance of a variant of the \kdoat problem in which each edge $(u,v)$ of the tree $T'$ has a known non-negative cost associated with it,  we have a shortcut for each pair of distinct vertices, and the function $c'$ determines the cost of each shortcut $(u,v)$ that can be added to the tree. Therefore, in our variant of \kdoat we are allowed to add the shortcut $(u,v)$ of cost $c'(u,v)$ even if $(u,v)$ is an edge of $T'$ (i.e., the cost of the shortcut $(u,v)$ may be different from the cost of the edge $(u,v)$ of $T'$). We observe that all the results discussed in the previous sections hold even for this generalized version of \kdoat.\footnote{This is because we can reduce the generalized \kdoat problem instance into a \kdoat instance in linear time by splitting each edge of $T'$ of cost $\chi$ into two edges, one of cost $0$ and the other one of cost $\chi$, to avoid the presence of shortcuts that are parallel to the tree edges. All the shortcuts that are incident to the added vertex used to split an edge of $T'$ have a sufficiently large cost which renders them useless.}

Let $B$ be the set of {\em branch} vertices of $T$, i.e., the internal vertices of $T$ having a degree greater than or equal to $3$. It is a folklore result that a tree with $\lambda$ leaves contains at most $\lambda-1$ branch vertices. Therefore, we have $|B| \leq \lambda-1$. 

The reduced instance $T'$ has a set $V'$ of  $\eta=2n^{\frac{1}{2k+2}}$ vertices defined as follows. $V'$ contains all the branch vertices $B$ plus $\eta-|B|$ vertices $x_1,\dots,x_{\eta-|B|}$ of $T$ that are computed using Gonzalez algorithm on input $G=T$ and $h=\eta-|B|$. By Lemma~\ref{lemma:gonzalez_fast_implementation}, the vertices $x_1,\ldots,x_{\eta-|B|}$ can be computed in $O(n+(\eta-|B|)^2\log n) = O(n)$ time. As $|B| \leq \lambda-1 = O\left(n^{\frac{1}{(2k+2)^2}}\right)$, it follows that $\eta-|B| >n^{\frac{1}{2k+2}}$.
The edges of $T'$ are defined as follows. There is an edge between two vertices $u,v \in V'$ iff the path $P$ in $T$ from $u$ to $v$ contains no vertex of $V'$ other than $u$ and $v$, i.e., $V(P) \cap V'=\{u,v\}$. The cost of an edge $(u,v)$ of $T'$ is equal to $\dist{T}{u}{v}$. Then, the cost function $c'$ of $T'$ is defined for every pair of vertices $u,v \in V'$, with $u \neq v$, and is equal to $c'(u,v)=c(u,v)$. 
Given the vertices $V'$, the $\eta-1$ costs of the edges $(u,v)$ of $T'$, that are equal to the values $\dist{T}{u}{v}$, can be computed in $O(n)$ time using a depth-first traversal of $T$ (from an arbitrary vertex of $T'$).

We use the exact algorithm of Corollary~\ref{cor:exact_algorithm} to compute an optimal solution $S'$ for the reduced instance in time $O(\eta^{2k+1}\cdot k \log \eta)=O\left(2^{2k+2}\cdot n^{1-\frac{1}{2k+2}}\cdot k \log n\right)=O(n)$.
The algorithm returns $S'$ as a solution for the original problem instance. We observe that the algorithm runs in $O(n)$ time.
In order to prove that the algorithm computes a $(1+\varepsilon)$-approximate solution, we first give a preliminary lemma showing that each vertex of $T$ is not too far from at least one of the vertices in $\{x_1,\ldots,x_{\eta-|B|}\}$. 
\begin{restatable}{lemma}{gonzalezreduced}
\label{lemma:gonzalez_reduced}
If $T$ has $\lambda=O\left(n^{\frac{1}{(2k+2)^2}}\right)$ leaves, then, for every $v \in V(T)$, there exists $i\in\{1,\ldots, \eta-|B|\}$ such that $\dist{T}{x_i}{v} \leq \frac{\varepsilon}{4(k+2)}D^*$.
\end{restatable}

\begin{restatable}{theorem}{ptaslineartime}
\label{thm:ptas_linear_time}
Let $\varepsilon > 0$ be a constant. 
Given a metric \kdoat instance with $k=o(\sqrt{\log n})$ and such that $T$ is a tree with $\lambda = O\big(n^{\frac{
1}{(2k+2)^2}}\big)$ leaves, the algorithm computes a $(1+\varepsilon)$-approximate solution in $O(n)$ time.
\end{restatable}
\begin{proof}
We already proved through the section that the algorithm runs in $O(n)$ time. So, it only remains to prove the approximation factor guarantee.

We define a function $\phi: V(T) \rightarrow \{x_1,\ldots,x_{\eta-|B|}\}$ that maps each vertex $v \in V(T)$ to its closest vertex $x_i$, with $i \in \{1,\ldots, \eta-|B|\}$ w.r.t.\ the distances in $T$, i.e., $\dist{T}{v}{\phi(v)} = \min_{1\leq i \leq \eta-|B|} \dist{T}{v}{x_i}$.

We now show that there exists a set $S$ of at most $k$ shortcuts such that (i) each edge $e \in S$ is between two vertices in $\{x_1,\ldots,x_{\eta-|B|}\}$ and (ii) $\diam(T'+S) \leq (1+\frac{k\varepsilon}{k+2})D^*$. The set $S$ is defined by mapping each shortcut $e=(u,v) \in S^*$ of an optimal solution for the original \kdoat instance to the shortcut $\big(\phi(u),\phi(v)\big)$ (self-loops are discarded). Clearly, (i) holds. To prove (ii), fix any two vertices $u$ and $v$ of $T'$. We first show that $\dist{T+S}{u}{v} \leq (1+\frac{k\varepsilon}{k+2})D^*$ and then prove that $\dist{T'+S}{u}{v} \leq (1+\frac{k\varepsilon}{k+2})D^*$. 

Let $P$ be a shortest path in $T+S^*$ between $u$ and $v$ and assume that $P$ uses the shortcuts $e_1=(u_1,v_1), \dots, e_t=(u_t,v_t) \in S^*$, with $t \leq k$. Consider the (not necessarily simple) path $P'$ in $T+S$ that is obtained from $P$ by replacing each shortcut $e_i$ with a {\em detour} obtained by concatenating the following three paths: (i) the path in $T$ from $u_i$ to $\phi(u_i)$; (ii) the shortcut $(\phi(u_i),\phi(v_i))$; (iii) the path in $T$ from $\phi(v_i)$ to $v_i$.

The overall cost of the detour that replaces the shortcut $e_i$ in $P'$ is at most $c(e_i)+ \frac{\varepsilon}{k+2}D^*$. Indeed,using that $c(\phi(u_i),\phi(v_i)) \leq \dist{T}{u_i}{\phi(u_i)} + c(e_i) + \dist{T}{\phi(v_i)}{v_i}$ together with Lemma~\ref{lemma:gonzalez_reduced} that implies $\dist{T}{u_i}{\phi(u_i)}, \dist{T}{\phi(v_i)}{v_i} \leq \frac{\varepsilon}{4(k+2)}D^*$, we obtain 
\begin{multline*}
\dist{T}{u_i}{\phi(u_i)} + c(\phi(u_i),\phi(v_i)) + \dist{T}{\phi(v_i)}{v_i}  \\
\leq 2\dist{T}{u_i}{\phi(u_i)} + c(e_i) + 2\dist{T}{\phi(v_i)}{v_i} 
                        \leq c(e_i) + \frac{\varepsilon}{k+2}D^*.
\end{multline*}
As a consequence, $c(P') \leq c(P)+\frac{t\varepsilon}{k+2} D^* \leq c(P)+\frac{k\varepsilon}{k+2}D^*$ and, since $c(P) \leq \diam(T+S^*) \leq D^*$, we obtain $c(P') \leq (1+\frac{k\varepsilon}{k+2})D^*$.

To show that $\dist{T'+S}{u}{v} \leq (1+\frac{k\varepsilon}{k+2})D^*$, i.e., that $\diam(T'+S) \leq (1+\frac{k\varepsilon}{k+2})D^*$, it is enough to observe that $P'$ can be converted into a path $P''$ in $T'+S$ of cost that is upper bounded by $c(P')$. More precisely, we partition the edges of $P'$ except those that are in $S$ into subpaths, each of which has two vertices in $\{x_1,\ldots,x_{\eta-|B|}\}$ as its two endvertices and no vertex in $\{x_1,\ldots,x_{\eta-|B|}\}$ as one of its internal vertices. The path $P''$ in $T'+S$ is defined by replacing each subpath with the edge of $T'$ between its two endvertices. Clearly, the cost of this edge, being equal to the distance in $T$ between the two endivertices, is at most the cost of the subpath. Therefore, the cost of $P''$ in $T'+S$ is at most the cost of $P'$ in $T+S$; hence $\dist{T'+S}{u}{v} \leq (1+\frac{k\varepsilon}{k+2})D^*$.

We conclude the proof by showing that the solution $S'$ computed by the algorithm satisfies $\diam(T+S') \leq (1+\varepsilon)D^*$. The solution $S'$ is  an optimal solution for the reduced instance. As a consequence, $\diam(T'+S') \leq \diam(T'+S) \leq (1+\frac{k\varepsilon}{k+2})D^*$. Let $u$ and $v$ be any two vertices of $T$. We have that $\dist{T'+S'}{\phi(u)}{\phi(v)} \leq (1+\frac{k\varepsilon}{k+2})D^*$. Moreover, by Lemma~\ref{lemma:gonzalez_reduced}, we have that $\dist{T}{u}{\phi(u)},\dist{T}{v}{\phi(v)} \leq \frac{\varepsilon}{4(k+2)}D^*$. Therefore,
\begin{align*}
\dist{T+S'}{u}{v}   & \leq \dist{T}{u}{\phi(u)} + \dist{T+S'}{\phi(u)}{\phi(v)} + \dist{T}{v}{\phi(v)} \\
                    & \leq \frac{\varepsilon}{4(k+2)}D^* + \dist{T'+S'}{\phi(u)}{\phi(v)} + \frac{\varepsilon}{4(k+2)}D^*\\
                    & \leq \frac{2\varepsilon}{4(k+2)}D^*+ \left(1+\frac{k\varepsilon}{k+2}\right)D^* < (1+\varepsilon)D^*.
\end{align*}
Hence $\diam(T+S') \leq (1+\varepsilon)D^*$. This completes the proof.
\end{proof}

\bibliography{bibliography.bib}

\newpage
\appendix

\section{Transforming \texorpdfstring{$T$}{T} into a binary tree}
 \label{sec:binarizzazione}
 
 In Section~\ref{sec:ds} we assumed that the input tree $T$ (rooted in some arbitrary vertex $r$) is binary. If this is not the case, we can preliminarily transform $T$ into a tree $T'$ that contains all the vertices of $T$ plus at most $O(n)$ additional vertices and ensures that the diameter of $T'+S$ coincides with the diameter of $T+S$. 
 
 In order to do that we start, with $T' = T$ and we iteratively reduce the degree of each vertex $v$ that has degree at least $3$ until there is no such vertex left. To reduce the degree of $v$ we remove all edges from $v$ to its children $v_1, v_2, \dots, v_h$, we add a full binary tree\footnote{A binary tree is full if all its internal nodes have degree $3$ except the root $r$.} with exactly $h$ leaves $u_1, u_2, \dots, u_h$ whose root coincides with $v$, and we add all the edges $(u_i, v_i)$ for $i=1, \dots, h$. We set the cost of the novel edges in the binary tree to $0$ and the cost of the edge $(u_i, v_i)$ to the original cost of edge $(v, v_i)$. See Figure~\ref{fig:binary} for an example.

 A vertex with degree $h \ge 3$ causes the addition of $2h-2$ new vertices, therefore the resulting tree $T'$ contains at most $2n-2$ more vertices w.r.t.\ $T$.
 The above transformation preserves all distances between vertices in $T$, does not increase the diameter (indeed, the eccentricity of a new vertex $u$ is at most the eccentricity of the vertex $v$ in $T$ that becomes the root of the binary tree containing $u$), and can be carried out in time $O(n)$.
 
  \begin{figure}[t]
     \centering
     \includegraphics[scale=1.35]{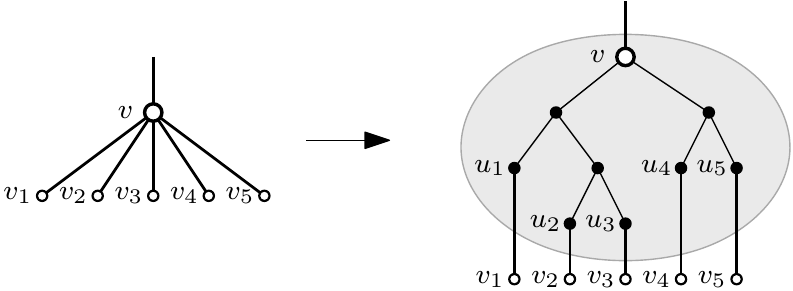}
     \caption{Transforming a generic tree $T$ into a binary tree $T'$. On the left side we have a vertex $v$ of $T$ with 5 children. On the right side we show the portion of the tree $T'$ in which all the edges $(v,v_i)$ of $T$, with $1 \leq i \leq 5$, have been replaced by a suitable gadget.}
     \label{fig:binary}
 \end{figure}

\section{Omitted proofs}
\label{app:omitted_proofs}

\lemmalbtreesgeneralized*
\begin{proof}
We extend the construction given in the proof of Lemma~\ref{lemma:lb_trees} to general values of $k=o(n)$.
We modify the input tree $T$ as follows. Let $Z \subseteq L_1$ be a set of size $|Z|=k-3$ and let $y \in L_4$ be fixed. The input tree $T$ is obtained from the input tree defined in the proof of Lemma~\ref{lemma:lb_trees} where each edge $(x_1,z)$, with $z \in Z$, is swapped with the edge $(y,z)$, i.e., we remove the edge $(x_1,z)$ and add the edge $(y,z)$. We set the cost of the new edge $c(y,z)$ to $9$ and notice that this coincides with the distance between $y$ and $z$ in $G_{a,b}$ and in $G$.
Similarly to the proof of Lemma~\ref{lemma:lb_trees}, we just need to prove the following three properties: (1) each instance $I_{a,b}$ admits a solution $S$ such that $\diam(T+S) \le 9$; (2) all solutions $S$ to $I$ are such that $\diam(T+S) \ge 10$; (3) if $u \neq a$ or $v \neq b$ then $\dist{G}{u}{v} = \dist{G_{a,b}}{u}{v}$.

To prove (1), we just need to observe that for any instance $I_{a,b}$, the set $S=\{(x_1,a),(a,b),(b,x_4)\} \cup \{(x_1,z) \mid z \in Z\}$ is such that $\diam(T+S)=9$.

Concerning (2), let us consider any solution $S$ of $k$ shortcuts. We prove if $S$ does not contain all the shortcuts $(z, x_1)$ with $z \in Z$, then $\diam(T+S) \ge 11$.
To see this, suppose that there is a vertex $z \in Z$ such that $(z, x_1) \not\in S$, and let $v \in L_4 \setminus \{y\}$ with no incident shortcuts in $S$. Then $\diam(T+S) \ge \dist{T+S}{v}{z} = \dist{T+S}{v}{x_4} +  \dist{T+S}{x_4}{z} \ge \dist{G}{v}{x_4} +  \dist{G}{x_4}{z} \ge 2 + 9 = 11$.

Notice that, once these $k-3$ shortcuts $(z,x_1)$ are forced to be in $S$, we are again in the situation of Lemma~\ref{lemma:lb_trees}, hence property (2) follows from the same arguments given in the proof of Lemma~\ref{lemma:lb_trees}. 

Finally, property (3) trivially holds since no distance in $G$ or $G_{a,b}$ changes as a consequence of our modifications to $T$. 

The rest of the proof is identical to that of Lemma~\ref{lemma:lb_trees}. \qed
\end{proof}

\thmfourapx*
\begin{proof}
By Lemma~\ref{lemma:gonzalez_fast_implementation}, the vertices $x_1,\ldots,x_{k+1}$ can be computed in $O(n+k^2\log n) = O(n)$ time. So, it only remains to prove that  $\diam(T+S) \leq 4D^*$. We show this by proving that $\dist{T+S}{u}{v} \leq 4D^*$ for every two vertices $u,v \in V(T)$. Let $x_i$ (resp., $x_j$) the vertex that is closest to $u$ (resp., $v$) in $T$. By Lemma~\ref{lemma:frati}, we have that $\dist{T}{x_i}{u},\dist{T}{x_j}{v} \leq D^*$. Moreover, as $\diam(T+S^*) \leq D^*$ and $c$ satisfies the triangle inequality, we have that $c(x_1,x_t) \leq D^*$ for every $1\leq t \leq k+1$. Therefore, the path in $T+S$ from $u$ to $v$ that goes through $x_i$ and $x_j$ and uses at most 2 shortcuts has a length of at most $\dist{T}{u}{x_i} + \dist{T+S}{x_i}{x_j} + \dist{T}{x_j}{v} \leq D^*+2D^*+D^* \leq 4D^*$. \qed
\end{proof}

\subsection{Proof of Lemma~\ref{lemma:gonzalez_reduced}}
\label{app:proof_thm_ptast}

In this section we prove Lemma~\ref{lemma:gonzalez_reduced}. We first provide a useful lemma which shows an upper bound to the diameter of $T$ that depends only on $k$ and on $D^*$.

\begin{lemma}
\label{lemma:ub_diameter_input_tree}
$\diam(T) \leq (3k+2)D^*$.
\end{lemma}
\begin{proof}
Fix any vertex $v \in V(T)$ and let $T_v$ be a shortest path tree of $T+S^*$ rooted at $v$. Let $X = E(T) \setminus E(T_v)$ be the set of edges of $T$ that are not contained in $T_v$. We have that $|X| \leq k$ as $T+S^*$ contains at most $n-1+k$ edges.
The diameter of $T$ can be upper bounded by the overall sum of the diameters of the connected components in $T_v-S^*$ plus the overall sum of the costs of all the edges in $X$. 

The key observation to prove the claimed bound on the diameter of $T$ is that any two vertices are at a distance of at most $2D^*$ in $T_v$ as $\diam(T+S^*)=D^*$. This implies that
$T_v-S^*$ contains at most $k+1$ connected components, each of which has a diameter of at most $2D^*$.  Moreover, using the fact that $c$ satisfies the triangle inequality, the cost of each  edge $(u,v) \in X$ is at most $D^*$ as $c(u,v)$ is the cost of a shortest path from $u$ to $v$, and we know that $u$ and $v$ are at a distance of at most $D^*$ in $T+S^*$. Therefore,
$\diam(T) \leq (k+1)\cdot 2D^* + k\cdot D^* = (3k+2)D^*$. \qed
\end{proof}

The bound on the diameter of $T$ proved in Lemma~\ref{lemma:ub_diameter_input_tree} allows us to prove Lemma~\ref{lemma:gonzalez_reduced}.
\gonzalezreduced*
\begin{proof}
Let ${\mathcal P} =\{P_1,\dots,P_{\lambda-1}\}$ be the $\lambda-1$ paths obtained by partitioning the edges of $T$ as follows. Consider a copy $\widetilde{T}$ of $T$. As long as $\widetilde{T}$ is not the empty graph, we take a path $P$ from a leaf of $\widetilde{T}$ towards its closest branch vertex (if such a branch vertex does not exist, then $P=\widetilde{T}$). We add $P$ to our collection $\mathcal{P}$, and we update $\widetilde{T}$ by removing all the edges of $P$ from $\widetilde{T}$ and all the singleton vertices.

Fix a path $P_i$ and let $l_i$ be an endvertex of $P_i$. We subdivide $P_i$ into $\frac{n^{1/(2k+2)}}{\lambda}$ {\em intervals}, each interval having a width equal to $\frac{\lambda\cdot \diam(P_i)}{n^{1/(2k+2)}}$.  If we combine Lemma~\ref{lemma:ub_diameter_input_tree} with inequality~\eqref{eq:assumption_on_n}, we can bound the width of an interval by 
\begin{equation}\label{eq:interval_width}
\frac{\lambda\cdot (3k+2)D^*}{n^{1/(2k+2)}} \leq \frac{\varepsilon\lambda\cdot (3k+2)D^*}{12\lambda (k+2)^2} < \frac{3\varepsilon\lambda\cdot (k+2)D^*}{12\lambda (k+2)^2} = \frac{\varepsilon}{4(k+2)}D^*.
\end{equation}
The $j$-th interval contains all vertices $v$ of $P_i$ such that
\[
\frac{(j-1)\cdot \lambda\cdot \diam(P_i)}{n^{1/(2k+2)}} \leq \dist{P_i}{l_i}{v} \leq \frac{j\cdot \lambda\cdot  \diam(P_i)}{n^{1/(2k+2)}}.
\]
We observe that each vertex of $T$ is contained in at least one interval. We also observe that the overall number of intervals is at most $n^{\frac{1}{2k+2}}$.
Let $x_1,\ldots,x_{\eta-|B|}$ be the subset of vertices of $T'$ computed by running Gonzalez algorithm on $T$. Let $D=\min_{1\leq i < j \leq \eta-|B|}\dist{T}{x_i}{x_j}$. 
Since $\eta-|B| >n^{\frac{1}{2k+2}}$,
we have that the number $\eta-|B|$ of vertices computed using Gonzalez algorithm is strictly larger than the overall number of intervals. Therefore, by the pigeonhole principle, $D$ is bounded by the width of the largest interval which, using inequality~\eqref{eq:interval_width}, has a width of at most $\frac{\varepsilon}{4(k+2)}D^*$. As a consequence, by Lemma~\ref{lemma:gonzalez}, for every $v \in V(T)$, there exists $i\in\{1,\ldots, \eta-|B|\}$ such that $\dist{T}{x_i}{v} \leq D \leq \frac{\varepsilon}{4(k+2)}D^*$. The claim follows. \qed
\end{proof}

\pagebreak

\section{An auxiliary data structure: implementation}
\label{app:ds}

In this section we show how to implement the data structure described in Section \ref{sec:ds}

Our data structure consists of:
\begin{itemize}
    \item An array that stores, for each vertex $v$ in $T$, whether $v$ is a terminal or a non-terminal vertex, along with the value $\alpha_v$ for terminal vertices.
    \item An oracle of size $O(n)$ that is able to report, in constant time, both the distance $\dist{T}{u}{v}$ and the number of edges of the path between any two vertices $u$,$v$ in $T$.
    \item An oracle of size $O(n)$ that can answer \emph{lowest common ancestor} \cite{HarelT84} and \emph{level ancestor} \cite{BerkmanV94, BenderFC04} queries on $T$ in constant time.
    Given two vertices $u$, $v$, a lowest common ancestor query returns $\lca_T(u,v)$. 
    Given a vertex $v$ and a non-negative integer $\ell$, a level ancestor query returns the ancestor of $v$ at hop-distance $\ell$ from $v$ in $T$.  
    \item A \emph{link-cut} tree \cite{SleatorT83} that stores the shrunk tree $T(M)$ w.r.t.\ the current set $M$ of terminal vertices. A link-cut tree maintains a dynamic forest under vertex additions, vertex deletions, edge insertions (link), and edge deletions (cut). Each operation requires time $O(\log m)$, where $m=O(n)$ is the number of nodes in the forest. 
    \item A \emph{top-tree} \cite{AlstrupHDLT05} $\T_M$ that stores the tree $T$, supports marking and unmarking vertices, and is able to report the closest marked ancestor of given query node in $O(\log n)$ time per operation. The marked vertices in $\T_M$ will be exactly the vertices in $T(M)$. 
    \item A \emph{top-tree}  $\T$ that stores a forest of (edge-weighted) trees under link and cut operations and that, given a vertex $v$, can report the eccentricity of $v$ in the unique tree of the forest that contains $v$.\footnote{
    \cite{AlstrupHDLT05} already shows how to maintain the diameter of each tree in $\T$. It is not difficult to modify the above query to also retrieve vertex eccentricities.
    Alternatively, one can employ the following black-box solution to report the eccentricity of $v$: find the diameter $D$ of the tree containing $v$, link $v$ with an auxiliary vertex $u$ via an edge $(u,v)$ of weight $D$, compute the new diameter $D'$ of the tree containing $v$, cut $(u,v)$ from $\T$, and return $D'-D$.} Each query and update operation requires $O(\log m)$ time, where $m=O(n)$ is the number of nodes in the forest.
\end{itemize}

All of the above components can be initialized in time $O(n)$. We now describe how the operations are implemented. 

\paragraph*{MakeTerminal($v$):}
We start by updating the type of vertex $v$ and setting $\alpha_v$ to $0$ in constant time.
Next, we update the link-cut tree storing the shrunk tree $T(M)$ to account for the new terminal vertex. Notice that, since $T$ is a binary tree, so is any shrunk version of $T$ as it contains all pairwise lowest common ancestors. We can assume that $v$ was not already in $T(M)$ as a Steiner vertex, otherwise there is nothing to do.

Following the update, $T(M)$ changes in one of the following three ways (see Figure~\ref{fig:shrunk_tree_addition}):
\begin{enumerate}[(i)]
    \item $v$ becomes a new leaf of $T(M)$ dangling from some existing vertex $x$.
    \item an edge $e$ of $T(M)$ is split by the new vertex $v$.
    \item an edge $e$ of $T(M)$ is split by a new Steiner vertex $y$, which will be the parent of $v$.
\end{enumerate}

 \begin{figure}[t]
     \centering
     \includegraphics[width=\textwidth]{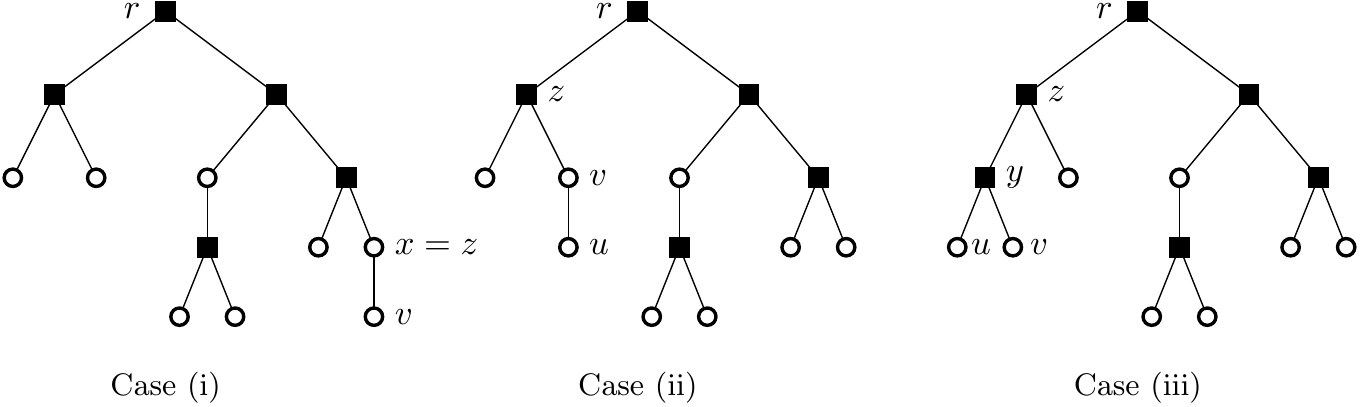}
     \caption{The three cases of updating the shrunk tree $T(M)$ after the call of MakeTerminal($v$). Case (i) shows the tree update when we only need to add $v$ as a leaf appended to some existing terminal vertex $x$. Case (ii) shows the tree update when $v$ splits an edge $(z,u)$ into two edges $(z,v)$ and $(v,u)$. Case (iii) shows the tree update when we need to add $v$ as a leaf appended to a new Steiner vertex $y$ that splits the edge $(z,u)$ into two edges $(z,y)$ and $(y,u)$. Case (iii) is the only case in which we also need to add a new Steiner vertex ($y$ in this example) that corresponds to the lowest common ancestor between $v$ and some other vertex ($u$ in this example). The tree $T$ is depicted in Figure~\ref{fig:shrunk_tree} and the vertex $v$ added in cases (i), (ii), and (iii) is the vertex $v_1$, $v_2$, or $v_3$ in Figure~\ref{fig:shrunk_tree}.}
     \label{fig:shrunk_tree_addition}
 \end{figure}

In order to distinguish the above cases, we query (in $O(\log n)$ time) $\T_M$ to find the closest marked ancestor $z$ of $v$. Next, we query (in constant time) the LCA oracle for the lowest common ancestors between $v$ and each of the (at most two) children of $z$ in $T(M)$. 
If all the LCA queries (possibly none) return $z$, then we are in case (i) with $x=z$ and we can simply add vertex $v$ to the link-cut tree and perform link operation to add the edge $(v,x)$.
Otherwise, there is some child of $u$ of $z$ in $T(M)$ such that $\lca(v,u) \neq z$. If $\lca(v,u) = v$ we are in case (ii) and the split edge is $(z,u)$. We delete $(z,u)$ from the link-cut tree, we add vertex $v$, and perform two link operations to add edges $(z,v)$ and $(v,u)$.
Finally, when there is some child of $u$ of $z$ in $T(M)$ such that $\lca(v,u) \not\in \{z,v\}$ we are in case (iii) with $y = \lca(v,u)$. We handle this case by cutting $(z,u)$ from the link-cut tree, inserting the new vertices $y$ and $v$, and performing $3$ link operations to add the edges $(z,y)$, $(y,u)$, and $(y,v)$.

In all of the above cases, we only perform a constant number of updates on the link-cut tree and hence the overall time spent is $O(\log n)$.

We conclude the operation by ensuring that the marked vertices of $\T_M$ are kept  up to date. We do so by marking $v$ and possibly $y$ (if we are in case (iii)) in time $O(\log n)$.

\paragraph*{Shrink():} We return (a copy of) the tree $T(M)$ maintained by the link-cut tree. The required time is linear in the size of $T(M)$, which is in $O(k)$.

\paragraph*{SetAlpha($v, \alpha$):} We simply set $\alpha_v$ to $\alpha$ in constant time.

\paragraph*{ReportFarthest():}

This operation works in three phases.

\textit{Phase 1.} In the first phase we propagate the quantities $\alpha_v$ from terminal vertices towards all vertices in $T(M)$.
In particular, given vertex $u$ in $T(M)$, we compute
$\beta_u = \min_{v \in M} \left( \alpha_v + \dist{T}{v}{u} \right)$ and we let $\nu_u$ denote the vertex $v$ for which the above minimum is attained.
Since the number of vertices of $T(M)$ is $O(|M|)$, all  the values $\beta_u$ and $\nu_u$ can be found in $O(|M|)$ time by performing a postorder visit of $T$ followed by a preorder visit of $T(M)$, while using the distance-reporting oracle to find the needed distances between vertices in $T$ in constant time. 

\textit{Phase 2.} The goal of the second phase is that of decomposing $\T$ into a forest that contains exactly one tree $T_u$ for each vertex $u$ in $T(M)$. Roughly speaking, the tree $T_u$ will contain all vertices $v$ that are ``closer'' to $u$ than to any other vertex $u'$ in $T(M)$, when the values $\alpha_u$ and $\alpha_{u'}$ are taken into account.

We now formalize the above intuition. We start by providing a tie-breaking scheme that will ensure that the sought partition is unique. 
To this aim we fix an arbitrary order of the vertices in $T$.
Given a vertex $u$ in $T(M)$ and a vertex $v$ in $T$ (possibly $u=v$), we define $\dtie{T}{u}{v}$ as the tuple $(\dist{T}{u}{v}, \ell, u)$, where $\ell$ is the number of edges in the path from $u$ to $v$ in $T$. We will also write $w + \dtie{T}{u}{v}$ as a shorthand for $(w + \dist{T}{u}{v}, \ell, u)$.
By comparing tuples lexicographically we obtain distance function that provides a total strict order relation over all possible pairs $(u,v)$. This distance preserves suboptimality of shortest paths and agrees with $\dist{T}{u}{v}$ in the sense that $\dist{T}{u}{v} < \dist{T}{u'}{v}$ implies $\dtie{T}{u}{v} < \dtie{T}{u'}{v}$.

We now consider each edge in $T(M)$. Any such edge $e=(u,v)$  corresponds to a unique simple path $P$ between the vertices $u$ and $v$ in $T$. 
Let $u=x_1, x_2, \dots, x_{|P|}=v$ be the vertices of $P$ as traversed from $u$ to $v$ and notice that, thanks to our distance reporting and level ancestor oracles, we know the number $|P|-1$  of edges in $P$ and we can access the generic $i$-th edge/vertex of $P$ in constant time.
Since $\dtie{T}{u}{x_i}$ and $\dtie{T}{v}{x_i}$ are monotonically increasing and monotonically decreasing w.r.t.\ $i$, respectively, we can binary search for the smallest index $i \ge 2$ such that $\beta_u + \dtie{T}{u}{x_i} > \beta_v + \dtie{T}{v}{x_i}$. Once such index $i$ has been found, we cut the edge $(x_{i-1}, x_i)$ from $\T$. See Figure~\ref{fig:bs_cut} for an example. The time needed to find all the $O(k)$ indices $i$ and to perform the corresponding cut operations is $O(k\log n)$ (i.e., $O(\log n)$ for each  index). 

\begin{figure}[t]
    \centering
    \includegraphics[scale=0.9]{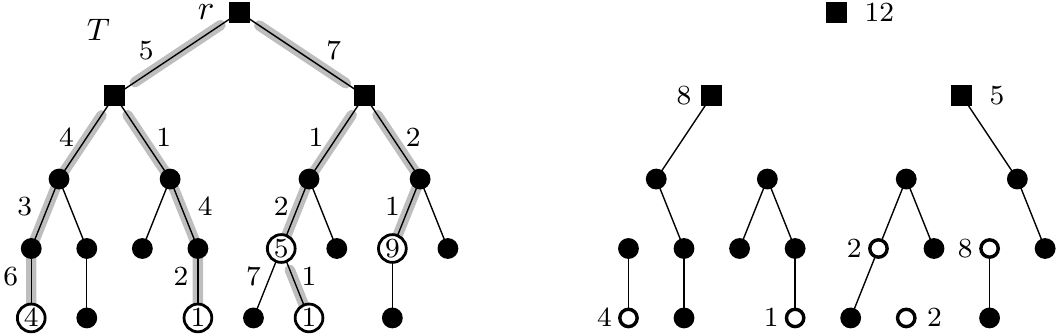}
    \caption{Left: The tree $T$ in which terminal vertices $v$ are depicted as white circles labelled with their value $\alpha_v$, while Steiner vertex are black squares. The paths $P$ of $T$ corresponding to edges in $T(M)$ are shaded.
    Right: the forest obtained after cutting the edges $(x_{i-1}, x_i)$ of each path $P$ from $\T$. Each tree contains exactly one vertex $u$ in $T(M)$. Each vertex $u$ in $T(M)$ is labelled with $\beta_u$.}
    \label{fig:bs_cut}
\end{figure}

After all the edges $e$ in $T(M)$ have been processed, $\T$ stores a forest with one tree $T_u$ per vertex $u$ in $T(M)$. Such a tree $T_u$ contains exactly the vertices $v$ of $T$ that satisfy $\beta_u + \dtie{T}{u}{v} < \beta_{u'} + \dtie{T}{u'}{v}$ for every vertex $u' \neq u$ in $\T$. 

\textit{Phase 3.} In the third and last phase we compute the answer to the query and restore the original state of the data structure, so that we are ready to handle future operations.

To answer the ReportFarthest() query, we query each of the $O(k)$ trees $T_u$ in $\T$ for the vertex $u^*$ that is farthest from $u$ in $T_u$. The total time needed is $O(k \log n)$, i.e., $O(\log n)$ per query.
If we consider a vertex $z \in \arg \max_{u \in T(M)} ( \beta_u + \dist{T}{u}{u^*} )$, we have that the sought answer is exactly $\beta_z + \dist{T}{z}{z^*}$ while the corresponding pair of vertices in $T$ is $(\nu_z, z^*)$.

Finally, we undo all the cut operations by re-linking all the edges cut from $T$. Since each link operation on $\T$ takes $O(\log n)$ time, and there are $O(k)$ edges in $T(M)$, the overall time needed is $O(k \log n)$.

\end{document}